\documentclass[a4paper]{article}
\usepackage[utf8]{inputenc} 
\usepackage[T1]{fontenc}  
\usepackage{graphicx}
\usepackage{amsmath,amsfonts,amssymb}
\usepackage{amsthm, thmtools}
\usepackage{enumerate}
\usepackage[capitalize]{cleveref}
\usepackage{booktabs}

\usepackage{color}

\theoremstyle{plain}
\newtheorem{theorem}{Theorem}[section]
\newtheorem{corollary}[theorem]{Corollary}
\newtheorem{lemma}[theorem]{Lemma}
\newtheorem{proposition}[theorem]{Proposition}
\theoremstyle{definition}

\theoremstyle{remark}

\declaretheorem[name=Example,style=definition,qed=$\triangle$,sibling=theorem,Refname={Example,Examples}]{example}

\newcommand{\Z}{{\mathbb{Z}}}
\newcommand{\F}{{\mathbb{F}}}

\newcommand{\vv}{{\bf  v}}

\newcommand{\vw}{{\bf  w}}

\newcommand{\vx}{{\bf  x}}
\newcommand{\vh}{{\bf  h}}

\newcommand{\vg}{{\bf  g}}

\newcommand{\bzero}{{\bar 0}}
\newcommand{\bone}{{\bar 1}}

\newcommand{\ow}{\overline{\omega}}

\newcommand{\K}{{\cal K}}
\newcommand{\zero}{{\mathbf{0}}}
\newcommand{\one}{{\mathbf{1}}}

\newcommand{\gh}{\operatorname{GH}}
\newcommand{\rank}{\operatorname{rank}}

\newcommand{\wt}{{\rm wt}}
\newcommand{\bfomega}{\boldsymbol{\omega}}

\newcommand{\GF}{\operatorname{GF}}
\definecolor{gold}{rgb}{0.85,.26,0}

\title{Rank and Kernel of  $\F_p$-Additive Generalised Hadamard Codes\thanks{This work was partially supported by the
Spanish MINECO under Grant TIN2016-77918-P, and
by the Catalan AGAUR under Grant 2017SGR-00463.\newline \hspace*{0.3cm} $^1$Steven T. Dougherty is with the Department of Mathematics, University of Scranton, Scranton PA 18510, USA. \newline \hspace*{0.3cm} $^2$Josep Rif\`{a} and Merc\`{e} Villanueva
are with the Department of Information and Communications
Engineering, Universitat Aut\`{o}noma de Barcelona, 08193 Cerdanyola del Vall\`{e}s, Spain.}}

\author{Steven T. Dougherty$^1$, Josep Rif\`{a}$^2$ and Merc\`{e} Villanueva$^2$}

\date{\today}

\begin{document}

\maketitle

\begin{abstract}
A subset of a vector space $\F_q^n$ is $K$-additive if it is a linear space over the subfield $K\subseteq \F_q$. Let $q=p^e$, $p$ prime, and $e>1$. Bounds on the rank and dimension of the kernel of generalised Hadamard ($\gh$) codes which are $\F_p$-additive are established. For specific ranks and dimensions of the kernel within these bounds, $\F_p$-additive $\gh$ codes are constructed. Moreover, for the case $e=2$, it is shown that the given bounds are tight and it is possible to construct an  $\F_p$-additive $\gh$ code for all allowable ranks and dimensions of the kernel between these bounds. Finally, we also prove that these codes are self-orthogonal with respect to the trace Hermitian inner product, and generate pure quantum codes. 
\end{abstract}

\section{Introduction}

Let $\F_q=\GF(q)$ denote the finite field with $q$ elements, where $q=p^e$, $p$ prime.
Let $\F^n_q$ be the vector space of dimension $n$ over $\F_q$.
The {\it Hamming distance} between vectors $\vw$, $\vv \in \F^n_q$,
denoted by $d(\vw,\vv)$, is the number of coordinates in which $\vw$ and
$\vv$ differ.
A {\it code} $C$ over $\F_q$ of length $n$ is a nonempty subset of $\F^n_q$. The
elements of $C$ are called {\it codewords}. The {\it minimum distance} of a code is the
smallest Hamming distance between any pair of distinct codewords. A code $C$ over $\F_q$ is called {\it
linear} if it is a linear space over $\F_q$ and, it is called {\it $K$\nobreakdash-additive}
if it is linear over a subfield $K \subseteq \F_q$.
The dimension of a $K$\nobreakdash-additive code $C$ over $\F_q$ is defined as the number $k$ such that
$q^k=|C|$. Note that $k$ is not necessarily an integer, but $ke$ is an integer, where $q=|K|^e$.

Two codes $C_1$, $C_2 \subset \F^n_q$ are said to be {\it permutation equivalent}
if there exists a permutation $\sigma$ of the $n$ coordinates such that
$C_2=\{\sigma(c_1,c_2,\ldots,c_n)=(c_{\sigma^{-1}(1)},\ldots, c_{\sigma^{-1}(n)}) : (c_1,c_2,\ldots,c_n) \in C_1 \}$. Without loss of generality, we
shall assume, unless stated otherwise, that the all-zero vector, denoted by $\zero$, is
in $C$.

Two structural parameters of (nonlinear) codes are the dimension of
the linear span and the kernel. The {\it linear span} of a code $C$ over $\F_q$, denoted by ${\cal R}(C)$, is the
subspace over $\F_q$ spanned by $C$, that is ${\cal R}(C)=\langle C \rangle$.
The dimension of ${\cal R}(C)$ is called the {\it rank} of $C$ and is denoted by $\rank(C)$.
If $q=p^e$, $p$ prime, we can also define ${\cal R}_p(C)$ and $\rank_p(C)$ as the 
$\F_p$-additive code spanned by $C$ and its dimension, respectively.
The {\it kernel} of a code $C$ over $\F_q$, denoted
by ${\cal K}(C)$, is defined as ${\cal K}(C) =\{\vx\in \F_q^n \ : \ \alpha \vx+C=C$ for all $\alpha \in \F_q  \}.$
If $q=p^e$, $p$ prime, we can also define the {\it $p$-kernel} of $C$ as
$\K_p(C)=\{ \vx\in \F_q^n \ : \ \vx+C=C \}.$  Since we assume that $\zero \in C$, then
$\K(C)$ is a linear subcode of $C$ and $\K_p(C)$ is an \mbox{$\F_p$\nobreakdash-additive subcode}.
We denote the dimension of the kernel (resp., $p$-kernel) of $C$ by $\ker(C)$ (resp., $\ker_p(C)$).
These concepts were first defined
in \cite{merce1} for codes over $\F_q$, generalising the binary case described previously in \cite{BGH}, \cite{PhLV}.
In \cite{merce1}, it was proved that any code $C$ over $\F_q$ can be written as the union of
cosets of $\K(C)$ (resp., $\K_p(C)$), and $\K(C)$ (resp., $\K_p(C)$) is the largest such linear code
over $\F_q$ (resp., $\F_p$) for which this is true. Moreover, it is clear that $\K(C) \subseteq \K_p(C)$.

\medskip
A {\it generalised Hadamard} ($\gh$) {\it matrix} $H(q,\lambda)=(h_{ij})$ of order $n=q\lambda$
over $\F_q$ is a $q\lambda \times q\lambda$ matrix with entries from
$\F_q$ with the property that for every $i,j$, $1\leq i<j\leq q\lambda$, each of the multisets
$\{ h_{is} - h_{js} :  1 \leq s \leq q\lambda \}$ contains every element of $\F_q$ exactly $\lambda$ times.
It is known that since $(\F_q,+)$ is an abelian group then $H(q,\lambda)^T$
is also a $\gh$ matrix, where $H(q,\lambda)^T$ denotes the transpose of $H(q,\lambda)$ \cite{jun}.
An ordinary Hadamard matrix of order $4\mu$ corresponds to a $\gh$ matrix $H(2,\lambda)$ over $\F_2$, where $\lambda=2\mu$.

Two $\gh$ matrices $H_1$ and $H_2$ of order $n$ are said to be {\it equivalent}
if one can be obtained from the other by a permutation of the rows and columns
and adding the same element of $\F_q$ to all the coordinates in a row or in a column.
We can always change the first row and column of a $\gh$ matrix
into zeros and we obtain an equivalent $\gh$
matrix which is called {\it normalized}. From a normalized $\gh$ matrix $H$,
we denote by $F_H$ the code over $\F_q$ consisting of the rows of $H$,
and $C_H$ the one defined as $C_H=\bigcup_{\alpha \in \F_q} (F_H+\alpha {\bf 1})$, where $F_H+\alpha {\bf 1}=\{ \vh+ \alpha {\bf 1} : \vh \in F_H\}$ and ${\bf 1}$ denotes the all-one vector.
The code $C_H$ over $\F_q$ is called {\it generalised Hadamard} ($\gh$) {\it code}. Note that $F_H$ and $C_H$ are generally nonlinear codes over $\F_q$.

To check whether two normalized $\gh$ matrices are equivalent
is known to be an NP-hard problem \cite{NPHard}. However, we can use the invariants related to the linear span and kernel
of the corresponding $\gh$ codes in order to help in their classification,
since if two $\gh$ codes have different ranks or dimensions of the kernel, the normalized $\gh$ matrices are nonequivalent. Given a normalized $\gh$ matrix $H$, to establish the rank and 
dimension of the kernel of the corresponding code $F_H$ is the same as to establish these values for the code $C_H$, since  
\begin{equation}\label{lemm:1}
\begin{aligned}
   \rank(C_H)= &\rank(F_H) +1 \quad \mbox{ and }\\
   \ker(C_H) = &\ker(F_H) +1
\end{aligned}
\end{equation} by \cite[Lemma 1]{DRV15}. In this paper, we focus on the codes $C_H$, although
everything could be rewritten in terms of the  codes $F_H$. 
It is important to emphasise that this is true as long as the $\gh$ matrix $H$ is normalized.

The rank and dimension of the kernel for ordinary Hadamard codes over $\F_2$ have already been studied.
Specifically, lower and upper bounds for these
two parameters were established, and the construction of an Hadamard code for all allowable ranks and
dimensions of the kernel between these bounds was given \cite{HadPower2,HadAnyPower}. The values of the rank and dimension of the kernel for $\Z_2\Z_4$-linear Hadamard codes were given  in \cite{HadAdditius}, and these invariants for $\Z_{2^s}$-linear Hadamard codes  have been studied in \cite{Carlos1,Carlos2}. The $\Z_2\Z_4$-linear Hadamard codes (resp. $\Z_{2^s}$-linear Hadamard codes) are the Hadamard codes over $\F_2$ obtained as the Gray map image of $\Z_2\Z_4$\nobreakdash-additive codes (resp. $\Z_{2^s}$-additive codes), which are subgroups of $\Z_2^\alpha \times \Z_4^\beta$ (resp. $\Z_{2^s}^\beta$). 

Some of the results on the rank and dimension of the kernel for Hadamard codes over $\F_2$ have been generalised to $\gh$ codes over $\F_q$ with $q\not =2$ \cite{DRV15}. Specifically, some lower and upper bounds for the dimension of the kernel, and for the rank once the dimension of the kernel is fixed, were given. Moreover constructions of $\gh$ codes having different values for these invariants within these bounds, were presented.
In this paper, we continue studying the rank and dimension of the kernel for $\gh$ codes over $\F_q$.
However, now we focus on a specific family of $\gh$ codes, namely the  {\it $\F_p$-additive $\gh$  codes}, that is,
$\F_p$-additive codes over $\F_q$ obtained from $\gh$ matrices $H(q,\lambda)$.

The paper is organized as follows. In \Cref{sec:bounds}, lower and upper bounds on the dimension of the kernel, and the rank once the dimension of the kernel is fixed, are given. In \Cref{sec:Kronecker}, several constructions of $\F_p$-additive $\gh$ 
codes over $\F_q$ with $q=p^e$, $p$ prime and $e>1$, are shown. In \Cref{sec:CombConst}, by using these constructions, we establish for which allowable pairs $(r,k)$, where $r$ is the rank and $k$ the dimension of the kernel, there exists an  $\F_p$-additive $\gh$ code having these invariants. Finally, in \Cref{sec:quantum}, we see that the $\F_p$-additive $\gh$ codes over $\F_{p^2}$ can be used to generate pure quantum codes since they are self-orthogonal with respect to the trace Hermitian inner product.

\section{Bounds on the rank and dimension of the kernel}
\label{sec:bounds}

In this section, we state new results on the rank and dimension of the kernel for
 $\F_p$-additive generalised Hadamard  codes. Note that a $\gh$ matrix $H(p,\lambda)$ over $\F_p$, $p$ prime,
generates an  $\F_p$-additive $\gh$  code $C_H$ of length $n=\lambda p = p^t$ if and only if
$\rank(C_H)=\rank_p(C_H)=\ker_p(C_H)=\ker(C_H)=1+t$. Therefore, we focus on  $\F_p$-additive $\gh$ codes over $\F_q$ with
$q=p^e$, $e>1$.

\begin{proposition}  \cite[Proposition 9]{DRV15}\label{bounds}
Let $H(q,\lambda)$ be a $\gh$ matrix over $\F_q$, where $q=p^e$, $p$ prime, and $e\geq 1$. Let $n=q\lambda=p^t s$ such that $\gcd(p,s)=1$.
Then $1 \leq \ker(C_H) \leq \ker_p(C_H) \leq 1+t/e$.
\end{proposition}

\begin{lemma} \cite[Lemma 16]{DRV15} \label{lemm:maxkernel}
Let $C_H$ be a $\gh$ code of length $n=q^h s$ over $\F_q$, where $s\not =1$ and $s$ is not a multiple of $q$. Then $\ker(C_H) \leq h$.
\end{lemma}

\begin{lemma}\label{lemm:3.1}
Let $H(q,\lambda)$ be a $\gh$ matrix over $\F_q$ such that $C_H$ is $\F_p$\nobreakdash-additive.
Let $n=q\lambda=p^t s$ such that $\gcd(p,s)=1$, where $q=p^e$, $p$ prime, and $e\geq 1$.
For any $v\in C_H$, $v\in \K(C_H)$ if and only if $\mu v \in C_H$ for all $\mu \in \F_q$.
\end{lemma}

\begin{proof}
Assume that $\mu v \in C_H$ for all $\mu \in \F_q$. Since $C_H$ is $\F_p$\nobreakdash-additive, for any $w\in C_H$ we have that $\mu v+w \in C_H$. Hence, the statement follows.
\end{proof}
\begin{proposition}\label{bounds-a}
Let $H(q,\lambda)$ be a $\gh$ matrix over $\F_q$, where $q=p^e$, $p$ prime, and $e>1$.
Let $n=q\lambda=p^t s$ such that $\gcd(p,s)=1$.
Then
\begin{enumerate}[(i)]
\item If $C_H$ is an $\F_p$\nobreakdash-additive code, then $s=1$.
\item The code $C_H$ is an $\F_p$\nobreakdash-additive code if and only if $$\rank_p(C_H)=\ker_p(C_H)=1+t/e.$$
\item If $C_H$ is an $\F_p$\nobreakdash-additive code  and $\ker(C_H)=k$, then $$ \frac{e+t-k}{e-1} \leq \rank(C_H)\leq 1+t-(e-1)(k-1).$$
\item If $C_H$ is an $\F_p$\nobreakdash-additive code, then $\rank(C_H)=\ker(C_H)= 1+t/e$ when $C_H$ is linear over $\F_q$ ($t$ is a multiple of $e$), or otherwise $$1\leq  \ker(C_H)\leq \lfloor t/e \rfloor.$$ 
\end{enumerate}
\end{proposition}

\begin{proof}
Since the number of codewords is $|C_H|=qn=p^{e+t} s$, if $C_H$ is $\F_p$\nobreakdash-additive, then we have that $s=1$ and
$\rank_p(C_H)=\ker_p(C_H)=1+t/e$. This proves items $(i)$ and $(ii)$.

Let  $C_H$ be an $\F_p$\nobreakdash-additive code with $\ker(C_H)=k$. The kernel $\K(C_H)$ is the largest linear subspace over $\F_q$ in $C_H$ such that $C_H$ can be partitioned into cosets of $\K(C_H)$. Specifically, there are  $|C_H|/q^k = q^{1+t/e}/q^{k}=q^{1+t/e-k}=p^{e+t-ek}$ cosets. Since $C_H$ and $\K(C_H)$ are linear over $\F_p$, the above cosets (that is, the elements of the  quotient $C_H / \K(C_H)$) have a linear structure over $\F_p$.
Therefore, there are $e+t-ek$ independent vectors over $\F_p$ generating these cosets, which means that the number of independent vectors over $\F_q$ generating these cosets  is upper bounded by $e+t-ek$. 
Hence  $\rank(C_H) \leq k+(e+t-ek) =1+t-(e-1)(k-1)$.
From \Cref{lemm:3.1}, for any $v\notin \K(C_H)$, the intersection of the linear space over $\F_q$ generated by $v$ and $C_H$ is, at most, of dimension $e-1$ over $\F_p$. 
Therefore, for the lower bound, we have that $\frac{e+t-ek}{e-1}+k= \frac{e+t-k}{e-1} \leq \rank(C_H)$ and item $(iii)$ follows.

For item $(iv)$, when $C_H$ is linear over $\F_q$, we have that $C_H=\K(C_H)$. Since $|C_H|= p^{e+t}=q^{1+t/e}$, $t$ is a multiple of $e$ and $\rank(C_H)=k=1+t/e$. Otherwise, $1\leq k < 1+t/e$ by \Cref{bounds} and item $(ii)$. In this case, if $t$ is a multiple of $e$, clearly $k\leq  \lfloor t/e \rfloor$. Finally, if $t$ is not a multiple of $e$, by \Cref{lemm:maxkernel}, since
$n=q^{\lfloor t/e \rfloor} p^{s'}$, where $1<p^{s'}<q$, we have that $k\leq \lfloor t/e \rfloor$.
\end{proof}

\begin{corollary} \label{coro:bounds-a}
Let $H(q,\lambda)$ be a $\gh$ matrix over $\F_q$, where $q=p^2$ and $p$ prime. If $C_H$ is an $\F_p$\nobreakdash-additive code of length $n=q\lambda=p^t$, then
\begin{enumerate}[(i)]
\item $\rank(C_H)+\ker(C_H) =2+t$.
\item If $2 \nmid t$, then $\rank(C_H)-\ker(C_H) \geq 3$.
\item If $2 \mid t$ and $C_H$ is nonlinear over $\F_q$, then $\rank(C_H)-\ker(C_H) \geq 2$.
\end{enumerate}
\end{corollary}
\begin{proof}
The first item is straightforward from item $(iii)$ in \Cref{bounds-a}.

For the second item, if $2=e \nmid t$, then $t=2h+1$. From item $(iv)$ in \Cref{bounds-a},
$\ker(C_H)\leq h$ and so $\rank(C_H)-\ker(C_H) \geq \rank(C_H)-h=2+t-\ker(C_H)-h\geq 2+t-2h=3$. For the third item, we can follow a similar argument, but considering that $t=2h$.
\end{proof}

\begin{example} \label{ex:e3}
For $q=p^3$, the second column in Table~\ref{tableF8bounds} gives all  possible values for the dimension of the kernel of $\F_p$-additive $\gh$  codes over $\F_{p^3}$ of length $n=p^t$ with $2\leq t \leq 12$. For each one of these values, the third column shows the possible values for the rank, given by \Cref{bounds-a}. 
\end{example}

\begin{table}[htp]%
\begin{center}
\begin{tabular}{cccc}
\toprule%
$t$ & $\ker(C_H)$& $\rank(C_H)$ & $\ker_p(C_H)=\rank_p(C_H)$ \\ \toprule
3   &   2        &  2   & 2\\
    &   1        &  {\bf 3,4} & 2\\ \hline 
4   &   1        &  {\bf 3,4},5   &7/3\\ \hline
5   &   1        & {\bf 4,5},6   &8/3\\   \hline
6   &   3        & 3   &3\\
    &   2         &  4,5  &3             \\ 
    &   1         &  {\bf 4,5,6},7 &3                \\ \hline
7   &   2      &    {\bf 4,5},6   &10/3\\
    &   1      &  {\bf 5,6,7},8 &10/3 \\ \hline
8   &   2      &  {\bf 5,6},7   &11/3\\
    &   1      &  {\bf 5,6,7,8},9 &11/3               \\   \hline
9   &   4       & 4 &4\\
    &   3        &  5,6 &4 \\
    &   2        & {\bf 5},6,7,8 &4\\ 
    &   1        & {\bf 6,7,8,9},10 &4 \\  \hline
10  &   3        & {\bf 5,6},7 & 13/3 \\
    &   2        & {\bf 6,7},8,9 &13/3 \\ 
    &   1        & {\bf 6,7,8,9,10},11 &13/3 \\  \hline  
11  &   3        &   {\bf 6,7},8 & 14/3\\
    &   2        &  {\bf 6,7,8},9,10 & 14/3\\ 
    &   1        &  {\bf 7,8,9,10,11},12 & 14/3 \\ \hline 
12   &   5       & 5 &5\\
    &   4        &  6,7 &5\\
    &   3        &  {\bf 6},7,8,9 &5\\ 
    &   2        & {\bf 7},8,9,10,11 &5  \\  
    &   1         & {\bf 7,8,9,10,11,12},13 &5\\ \bottomrule
\end{tabular}\caption{Parameters $\ker(C_H)$ and $\rank(C_H)$ for all  $\F_p$-additive $\gh$ codes $C_H$ over $\F_{p^3}$
of length $n=p^t$ with $3\leq t\leq 12$.}\label{tableF8bounds}
\end{center}
\end{table}

\section{Kronecker and switching constructions}
\label{sec:Kronecker}

In this section, we show that by using the Kronecker sum construction from $\F_p$-additive $\gh$ codes, we also obtain $\F_p$-additive $\gh$ codes.
Moreover, we present a switching construction that allows for the production of $\F_p$-additive $\gh$ codes. For all these constructions, we establish the values of the rank and dimension of the kernel for the obtained codes.

\medskip
A standard method to construct $\gh$ matrices from other $\gh$ matrices
is given by the {\it Kronecker sum construction} \cite{seo}, \cite{shr}. That is,
if $H(q,\lambda)=(h_{ij})$ is any $q\lambda \times q\lambda$ $\gh$ matrix over $\F_q$, and $B_1, B_2,
\ldots, B_{q\lambda}$ are any $q\mu\times q\mu$ $\gh$ matrices over $\F_q$, then the
matrix in Table~\ref{kron} gives a $q^2\lambda \mu\times q^2\lambda \mu$ $\gh$ matrix over $\F_q$, denoted by $H \oplus [B_1,
B_2, \ldots, B_n]$, where $n=q\lambda$. If $B_1=B_2=\cdots =B_n=B$, then we write $H \oplus [B_1, B_2,
\ldots, B_n]=H \oplus  B$.

\begin{table}[ht]
 \centering $ H \oplus [B_1,
B_2, \ldots, B_n]=\left(%
\begin{array}{cccc}
  h_{11}+B_1 & h_{12}+B_1 & \cdots & h_{1n}+B_1 \\
  h_{21}+B_2 & h_{22}+B_2 & \cdots & h_{2n}+B_2 \\
  \vdots & \vdots & \vdots & \vdots \\
  h_{n1}+B_n & h_{n2}+B_n & \cdots & h_{nn}+B_n \\
\end{array}%
\right)$
\caption{\label{kron} Kronecker sum construction}
\end{table}

Let $S_q$ be the normalized $\gh$ matrix $H(q,1)$ given by the multiplicative
table of $\F_q$. As for ordinary Hadamard matrices over $\F_2$, starting from a $\gh$ matrix $S^1=S_q$, we can
recursively define $S^h$ as a $\gh$ matrix $H(q,q^{h-1})$, constructed as $S^h=S_q \oplus
[S^{h-1},S^{h-1},\ldots,S^{h-1}]=S_q \oplus S^{h-1}$ for $h > 1$,
which is called a {\it Sylvester $\gh$ matrix}. Note that the corresponding $\gh$ code $C_{S^h}$ is linear
over $\F_q$, so $\rank(C_{S^h})=\ker(C_{S^h})=1+h$, by \Cref{lemm:1} or item $(iv)$ of \Cref{bounds-a}.

Now, we recall some known results on the rank and dimension of the kernel for $\gh$ codes constructed by using the Kronecker sum construction. In these cases, starting with $\F_p$-additive $\gh$ codes, we obtain $\F_p$-additive $\gh$ codes.

\begin{lemma}\label{lem:Kro1} 
Let $H_1$ and $H_2$ be two $\gh$ matrices over $\F_q$ and $H=H_1 \oplus H_2$. 
Then $\rank(C_H)=\rank(C_{H_1})+\rank(C_{H_2})-1$ and $\ker(C_H)=\ker(C_{H_1})+\ker(C_{H_2})-1$. Moreover, if $C_{H_1}$ and $C_{H_2}$ are $K$\nobreakdash-additive, then $C_{H}$ is also $K$\nobreakdash-additive and $\rank_p(C_H)=\rank_p(C_{H_1})+\rank_p(C_{H_2})-1$.
\end{lemma}

\begin{proof}
Straightforward from the proof of Lemma 3 in  \cite{DRV15}.
\end{proof}

\begin{corollary}\label{coro:2.2}Let $B$ be a $\gh$ matrix over $\F_q$ and $H=S_q \oplus B$.
Then  $\rank(C_H)=\rank(C_B)+1$ and $\ker(C_H)=\ker(C_B)+1$.
Moreover, if $C_B$ is $\F_p$-additive, then $C_H$ is also $\F_p$-additive and $\rank_p(C_H)=\rank_p(C_{B})+1$.
\end{corollary}

\begin{proof}
Straightforward from the proof of Corollary 4 in \cite{DRV15}.
\end{proof}

\medskip
Several switching constructions have been used to construct perfect 
codes in \cite{merce2}, \cite{merce3} ordinary Hadamard codes over $\F_2$ in \cite{HadPower2}, \cite{HadAnyPower}, and generalised Hadamard codes over $\F_q$ in \cite{DRV15}. In this paper, we present different constructions, based on this technique, in order to obtain $\F_p$-additive $\gh$ codes with different ranks and dimensions of the kernel. The first switching construction, given by \Cref{constructionSwitch0}, allows us to construct $\F_p$\nobreakdash-additive $\gh$ codes over $\F_{p^e}$ of length $n=p^{2e}$ with kernel of dimension 2 and rank $4$.

\begin{proposition}[Switching Construction I]\label{constructionSwitch0}
For $q=p^e$, $p$ prime, and any $e>1$,
there exists a $\gh$ matrix $H(p^e, p^{e})$ such that $C_H$  is an $\F_p$\nobreakdash-additive code over
$\F_{p^e}$ of length $n=p^{2e}=q^2$ with $\ker(C_H)=2$ and
$\rank(C_H)=4$.
\end{proposition}

\begin{proof}
Let $\zero$, $\one$, $\bfomega^{(1)}$, $\ldots$, $\bfomega^{(q-2)}$ be the elements $0, 1,
\omega, \ldots, \omega^{q-2}$ repeated $q$ times, respectively, where $\omega$ is a primitive element in $\F_q$.
Let $S^2=S_q \oplus S_q$ be the Sylvester $\gh$ matrix $H(q,q)$.
We can assume without loss of generality that $S^2$ is generated over $\F_q$ by the row vectors $\vv_1$ and $\vv_{2}$ of length $n=q^2$,
where
\begin{equation*}
\begin{split}
&\vv_1=(0, 1, \omega^{1}, \ldots, \omega^{q-2}, \ldots, 0, 1, \omega^{1}, \ldots, \omega^{q-2} ),  \ \ \textrm{and}\\
&\vv_2=(\zero, \one, \bfomega^{(1)}, \ldots, \bfomega^{(q-2)}).
\end{split}
\end{equation*}

Let $K$ be the linear subcode of $S^2$ generated by the row vector $\vv_2$.
The rows of $S^2$ can be partitioned into $q$ cosets of $K$, that is, $S^2=\cup_{\beta\in \F_q} (K+\beta \vv_1)$.
Let $\beta_e \in \F_p$ be the last coordinate of the element $\beta \in \F_q$ represented
as a vector from $\F_p^e$. Then we construct the matrix
$$H=(S^2\backslash \bigcup_{\substack{\beta\in \F_q,\\ \beta_e\not =0}} (K+\beta \vv_1)) \cup \bigcup_{\substack{\beta\in \F_q,\\ \beta_e \not =0}} (K+\beta \vv_1 + \beta_e\vg)=\bigcup_{\beta\in \F_q} K_\beta,$$
where $K_\beta = K+\beta \vv_1 + \beta_e \vg$ and $\vg=( 0, \ldots, 0, 0,1,\omega^{1}, \ldots, \omega^{q-2})$.

It is easy to see that $H$ is a $\gh$ matrix and $C_H$ is an $\F_p$-additive code.
Indeed, note that  $K_\beta + K_\gamma = K_{\beta+\gamma}$ for all $\beta, \gamma \in \F_q$.
Moreover, clearly, $\rank(F_H)=2+1=3$ and $K\subseteq \K(F_H)$.
It is also easy to prove that $K=\K(F_H)$.
Therefore, $\ker(F_H)=1$. By \Cref{lemm:1}, $\ker(C_H)=2$ and $\rank(C_H)=4$.
\end{proof}

\begin{example}
We construct a $\gh$ matrix $H(2^2,2^2)$ such that $C_H$ is an $\F_2$\nobreakdash-additive code
over $\F_{2^2}$ of length $n=2^4$ with $\ker(C_H)=2$ and $\rank(C_H)=4$.
We start with the $\gh$ matrix $S^2=S_4 \oplus S_4$, which is linear over $\F_{2^2}$ and is generated by
$\vv_1=(0, 1, \omega, \omega^2, 0, 1, \omega, \omega^2, 0, 1, \omega, \omega^2, 0, 1, \omega, \omega^2)$ and
$\vv_2=(0, 0, 0, 0, 1, 1, 1, 1, \omega,$ $\omega, \omega, \omega, \omega^2, \omega^2, \omega^2, \omega^2)$,
where $\omega$ is a primitive element in $\F_{2^2}$ and $\omega^2=\omega+1$.
Let $K=\langle \vv_2 \rangle$. By the proof of \Cref{constructionSwitch0}, we obtain the $\gh$ matrix
$H(2^2,2^2)=K \cup (K+\vv_1) \cup (K+\omega \vv_1+\vg) \cup (K+\omega^2 \vv_1+\vg)$,
where $\vg=(0,0,0,0,0,0,0,0,0,0,0,0,0,1,\omega,\omega^2)$, that is, the matrix
\begin{equation}
{\footnotesize
\left(
\begin{array}{cccccccc cccccccc}
0&0&0&0&0&0&0&0 &0&0&0&0&0&0&0&0\\
0&0&0&0&1&1&1&1 &\omega&\omega&\omega&\omega&\omega^2&\omega^2&\omega^2&\omega^2\\
0&0&0&0&\omega&\omega&\omega&\omega&\omega^2&\omega^2&\omega^2&\omega^2&1&1&1&1\\
0&0&0&0&\omega^2&\omega^2&\omega^2&\omega^2&1&1&1&1 &\omega&\omega&\omega&\omega\\

0&1&\omega&\omega^2& 0&1&\omega&\omega^2   &0&1&\omega&\omega^2& 0&1&\omega&\omega^2\\
0&1&\omega&\omega^2& 1&0&\omega^2& \omega  &\omega&\omega^2&0&1&\omega^2&\omega&1&0\\
0&1&\omega&\omega^2& \omega&\omega^2& 0&1  &\omega^2&\omega& 1&0&1&0&\omega^2&\omega\\
0&1&\omega&\omega^2& \omega^2&\omega&1&0   &1& 0&\omega^2&\omega&\omega&\omega^2&0&1\\

0&\omega&\omega^2&1 &0&\omega&\omega^2&1 &0&\omega&\omega^2&1  &0&\omega^2&1&\omega\\
0&\omega&\omega^2&1 &1&\omega^2&\omega&0 &\omega&0&1&\omega^2  &\omega^2&0&\omega&1\\
0&\omega&\omega^2&1 &\omega&0&1&\omega^2 &\omega^2&1&0&\omega  &1&\omega&0&\omega^2\\
0&\omega&\omega^2&1 &\omega^2&1&0&\omega &1&\omega^2&\omega&0  &\omega&1&\omega^2&0\\

0&\omega^2&1&\omega &0&\omega^2&1&\omega &0&\omega^2&1&\omega &0&\omega&\omega^2&1\\
0&\omega^2&1&\omega &1&\omega&0&\omega^2 &\omega&1&\omega^2&0 &\omega^2&1&0&\omega\\
0&\omega^2&1&\omega &\omega&1&\omega^2&0 &\omega^2&0&\omega&1 &1&\omega^2&\omega&0\\
0&\omega^2&1&\omega &\omega^2&0&\omega&1 &1&\omega&0&\omega^2 &\omega&0&1&\omega^2
\end{array}
\right).}
\end{equation}
\end{example}

The switching construction given by \Cref{constructionSwitch0} can be generalised to the case $H(p^e, p^{(h-1)e})$ with $h>1$, that is, when $t$ is a multiple of $e$. A first generalisation is shown in \Cref{constructionSwitch}, and a second one in \Cref{constructionSwitch2}.

\begin{proposition}[Switching Construction II]\label{constructionSwitch} 
For $q=p^e$, $p$ prime, and any $e>1$, $h>1$,
there exists a $\gh$ matrix $H(p^e, p^{(h-1)e})$ such that $C_H$  is an $\F_p$\nobreakdash-additive code over
$\F_{p^e}$ of length $n=p^{he}=q^h$ with $\ker(C_H)=h$ and
$\rank(C_H)=r$ for all $r \in \{h+2,\ldots,h+e\}$.
\end{proposition}

\begin{proof}
Let $S^h$ be the Sylvester $\gh$ matrix $H(q,q^{h-1})$.
We can assume without loss of generality that $S^h$ is generated by the vectors $\vv_1, \ldots, \vv_{h}$ of length $n=q^h$,
where 
$$\vv_i=(\zero_{i}, \one_{i}, \bfomega_{i}^{(1)}, \ldots, \bfomega_{i}^{(q-2)}, \ldots, \zero_{i}, \one_{i}, \bfomega_{i}^{(1)}, \ldots, \bfomega_{i}^{(q-2)} ),$$
$\zero_{i}$, $\one_{i}$, $\bfomega_{i}^{(1)}$, $\ldots$, $\bfomega_{i}^{(q-2)}$ are the elements $0, 1, \omega, \ldots, \omega^{q-2}$ repeated $q^{i-1}$ times, respectively, and $\omega$ is a primitive element in $\F_q$, for all $i\in \{1,\ldots,h\}$. All vectors $\vv_1,\ldots, \vv_h$ have length $q^h$ and are linearly independent over $\F_q$. The corresponding $\gh$ code $C_{S^h}$ is linear over $\F_q$, so $\rank(C_{S^h})=\ker(C_{S^h})=1+h$.

Let $K$ be the linear subcode of $S^h$ generated by the vectors $\vv_2,\ldots, \vv_h$.
Note that all $n=q^h$ coordinates are naturally divided into $q^{h-1}$ groups of
size $q$, which will be referred to as blocks, such that the
columns of $K$ in a block coincide. Moreover, the rows of $S^h$ can be partitioned into $q$ cosets of $K$, that is, $S^h=\cup_{\beta\in \F_q} (K+\beta \vv_1)$.

Now, consider the vectors $\vg_1, \ldots, \vg_{e-1} \in \F_q^{n}$, where $\vg_j$ has exactly the values $0,1,\omega,\ldots, \omega^{q-2}$ in the coordinate positions $(jq+1), \ldots, (j+1)q$, respectively, and  zeros elsewhere, for all $j \in \{1,\ldots,e-1\}$. Note that, for each $j\in \{1,\ldots,e-1\}$, these $q$ coordinate positions correspond to a block. Moreover, there are always enough blocks since $e \leq q \leq q^{h-1}=p^{e(h-1)}$.  
 Let $\beta \in \F_q$ be $\beta=(\beta_0,\ldots,\beta_{e-1})$ represented as a vector in $\F_{p}^e$. Then, we construct the matrix
\begin{equation}\label{eq:3.7}
    H^{(s)}=\bigcup_{\beta\in \F_q} K_\beta,
\end{equation}
where $s \in \{1,\ldots,e-1\}$ and $K_\beta = K+\beta \vv_1 + \sum_{j=1}^{s} \beta_j \vg_j$.

Next, we prove that the corresponding code $C_{H^{(s)}}$ is an $\F_p$-additive $\gh$ code for all $s\in \{1,\ldots,e-1\}$.
First of all, the length of $H^{(s)}$ is $q^h$ and the number of rows is also $q^h$. The code $C_{H^{(s)}}$  is $\F_p$\nobreakdash-additive over
$\F_{q}$, since $K_{\beta+\gamma}=K_\beta+K_\gamma$  for any $\beta, \gamma \in \F_q$. 
Finally, it can be seen easily that $H^{(s)}$ is a $\gh$ matrix 
by computing the differences between any two different rows.

By \Cref{lemm:3.1}, it is  straightforward to show that $\K(H^{(s)})=K$, so we have that
 $\ker(C_{H^{(s)}})=h-1+1=h$ by \Cref{lemm:1}. It is easy to see that all linearly independent vectors $\vv_1, \ldots, \vv_{h}$ from $S^h$ are in  $\langle C_{H^{(s)}} \rangle$. Moreover, the linearly independent vectors $\vg_1,\ldots, \vg_{s}$
 are also in $\langle C_{H^{(s)}} \rangle$. Hence, by \Cref{lemm:1}, $\rank(C_{H^{(s)}})=h+s+1$, which covers all values in the range of $h+2$ to $h+e$.
\end{proof}

We can make an slight modification in the proof of \Cref{constructionSwitch} allowing the construction of $\gh$ matrices $H(p^e,p^{(h-1)e})$ with $h>1$, where the dimension of the kernel $k$ of the corresponding codes ranges from $2$ to $h$ and for each one of these values the rank takes any value from $2h-k+2$ to $h+1+(h-k+1)(e-1)$. Note that this switching construction includes the previous two switching constructions I and II
given by \Cref{constructionSwitch0} and \Cref{constructionSwitch}, respectively. 

\begin{proposition}[Switching Construction III]\label{constructionSwitch2} 
For $q=p^e$, $p$ prime, and any $e>1$, $h>1$,
there exists a $\gh$ matrix $H(p^e, p^{(h-1)e})$ such that $C_H$  is an $\F_p$\nobreakdash-additive code over $\F_{p^e}$ of length $n=p^{he}=q^h$ with $\ker(C_H)=k$ and
$\rank(C_H)=r$, for all $k \in \{2, \ldots, h\}$ and 
$$r \in \{2h-k+2,\ldots,1+t-(e-1)(k-1)\}.$$
\end{proposition}

\begin{proof}
We consider the vectors $\vv_1, \ldots, \vv_{h}$ of length $n=q^h$ defined in the proof of \Cref{constructionSwitch}. Let $K$ be the linear subcode of $S^h$ generated by the vectors $\vv_3,\ldots, \vv_h$. In this case, the rows of $S^h$ can be partitioned into $q^2$ cosets of $K$, that is, $S^h=\cup_{\beta, \gamma\in \F_q} (K+\beta \vv_1+\gamma \vv_2 )$.
Let $\beta, \gamma \in \F_q$ be $\beta=(\beta_0,\ldots,\beta_{e-1})$ and $\gamma=(\gamma_0,\ldots,\gamma_{e-1})$, respectively, represented as vectors in $\F_{p}^e$.  Then, we construct the matrix
\begin{equation}\label{eq:3.8}
    H^{(s_1,s_2)}=\bigcup_{\beta,\gamma\in \F_q} K_{\beta,\gamma},
\end{equation}
where $s_1,s_2 \in \{1,\ldots,e-1\}$, $K_{\beta,\gamma} = K + \beta \vv_1 + \sum_{j=1}^{s_1} \beta_j \vg_j^{(1)} +\gamma \vv_2 + \sum_{j=1}^{s_2} \gamma_j \vg^{(2)}_j$. We take the vector $\vg_j^{(1)}=\vg_j$ defined as in the proof of \Cref{constructionSwitch} for $j\in \{1,\ldots,e-1\}$. 
The vector $\vg^{(2)}_j$ has exactly the values $\zero_{2}$, $\one_{2}$, $\bfomega_{2}^{(1)}$, $\ldots$, $\bfomega_{2}^{(q-2)}$ in the coordinate positions $(jq^2+1), \ldots, (j+1)q^2$, respectively, and  zeros elsewhere, for all $j\in \{1,\ldots, e-1\}$.

We can repeat again and again the above construction taking in the $z$-th round vectors $\vg_j^{(z)}$ with exactly the values $\zero_{z}$, $\one_{z}$, $\bfomega_{z}^{(1)}$, $\ldots$, $\bfomega_{z}^{(q-2)}$ in the coordinate positions $(jq^z+1), \ldots, (j+1)q^z$, respectively, and  zeros elsewhere, for all $j\in \{1,\ldots, e-1\}$. Again, there are always enough coordinates since $e \leq q=p^e$. We can follow this process until we consider the linear subcode $K=\langle \vv_h \rangle$, and we obtain the matrix $H^{(s_1,s_2,\ldots,s_{h-1})}$.

Next, we prove that the matrix $H=H^{(s_1,s_2,\ldots,s_{h-k+1})}$ is a $\gh$ matrix, and the corresponding code is $\F_p$-additive. First of all, the length and the number of rows of $H$ is $q^h$. The corresponding code $C_H$ is an $\F_p$\nobreakdash-additive code over $\F_{q}$ as in the proof of \Cref{constructionSwitch}. Finally, it can be easily seen that $H$ is a $\gh$ matrix by computing the differences between any two different rows.

Again, by \Cref{lemm:3.1}, it is easy to see that $\K(H)=K$, so we have that  $\ker(C_H)=h-(h-k+2)+1=k-1+1=k$ by \Cref{lemm:1}. For the $\rank(C_H)$, it is easy to see that all linearly independent vectors in $C_{S^h}$ are also in $\langle C_H \rangle$. Apart from that vectors, we also find in $\langle C_H \rangle$ the linearly independent vectors $\vg_1^{(1)},\ldots, \vg_{s_1}^{(1)}, \ldots, \vg_1^{(h-k+1)},\ldots,\vg_{s_{h-k+1}}^{(h-k+1)}$. Hence, $\rank(C_H)=h+1+s_1+\cdots+s_{h-k+1}$, which covers all values in the range from $h+1+(h-k+1)=2h-k+2$ to $h+1+(h-k+1)(e-1)=1+t-(e-1)(k-1)$.
\end{proof}

\section{New constructions with kernel of dimension 1}
\label{sec:constKernel1}

In this section, two new constructions of $\F_p$-additive $\gh$ codes having a kernel of dimension $1$, one with maximum rank and another one with minimum rank, are presented. In \Cref{sec:CombConst}, these constructions together with the Kronecker and switching constructions presented in \Cref{sec:Kronecker} will be used to construct $\F_p$-additive $\gh$ codes having different ranks and dimensions of the kernel.
 
\medskip
First, we introduce a new construction of $\gh$ matrices which allows us to guarantee that the obtained code $C_H$ of length $n=p^t$ is $\F_p$\nobreakdash-additive over $\F_{p^e}$, has kernel of minimum dimension 1 and maximum rank $t+1$.

\begin{proposition}\label{construction_t}
For $q=p^e$, $p$ prime, and any $t >e>1$,
there exists a $\gh$ matrix $H(p^e, p^{t-e})$ such that $C_H$  is an $\F_p$\nobreakdash-additive code over $\F_{p^e}$ of length $n=p^t$ with $\ker(C_H)=1$ and $\rank(C_H)=t+1$.
\end{proposition}

\begin{proof}
Let $S_{p^t}=H(p^t,1)$ be the $\gh$ matrix,  given by the multiplicative table of $\F_{p^t}$, that is, the matrix
\begin{equation}\label{taula}
H(p^t,1)=\left(
\begin{array}{cccccc}
0&0&0&\ldots&0&0\\
0&1&\omega&\ldots&\omega^{p^t-3}&\omega^{p^t-2}\\
0&\omega&\omega^2&\ldots&\omega^{p^t-2}&1\\
\ldots&\ldots&\ldots&\ldots&\ldots&\ldots\\
0&\omega^{p^t-2}&1&\ldots&\omega^{p^t-4}&\omega^{p^t-3}
\end{array}
\right),
\end{equation}
where $\omega$ is a primitive element in $\F_{p^t}$. Let $b_0+b_1x+\cdots+b_{t-1}x^{t-1} -x^t\in \F_{p}[x]$ be the primitive polynomial of $\omega$ and note that $b_0\not= 0$.
In this case, $C_H$ is a linear code over $\F_{p^t}$ and an $\F_p$\nobreakdash-additive code. By \Cref{bounds-a}, we have that $\rank_p(C_H)=\ker_p(C_H)=\rank(C_H)=\ker(C_H)=2$.

Now, for any $e$, $1<e<t$, we consider the projection map from $\F_{p^t}$ to $\F_{p^e}$
given by $$\vv=(v_1,\ldots,v_e,v_{e+1},\ldots,v_t)\in \F_{p^t} \longrightarrow \overline{\vv}=(v_1,\ldots,v_e) \in \F_{p^e}.$$ Note that
we can consider that the projection of $\omega \in \F_{p^t}$ gives a primitive element $\overline{\omega} =\alpha \in \F_{p^e}$.
Let $H_e$ be the matrix obtained from $H$ after changing each entry $\vv$ by $\overline{\vv}$. Since in any row of $H$ there are all the elements in $\F_{p^t}$, it is easy to see that in any row of $H_e$ there will be all the elements in $\F_{p^e}$, but each one repeated $\lambda=p^{t-e}$ times. The same happens taking the difference of any two different rows in $H$. Hence, $H_e(p^e,p^{t-e})$ is a $\gh$ matrix. Since $C_H$ is an $\F_p$\nobreakdash-additive code it is easy to see, by construction, that $C_{H_e}$ is also an  $\F_p$\nobreakdash-additive code.

Let $H_e^{(r)}$ be the matrix $H_e$ after removing the first row and column. To show that the dimension of the kernel $\ker(F_{H_e})$ is zero (or equivalently, $\ker(C_{H_e})=1$), we begin by noting that in any row of $H_e^{(r)}$ the only possible consecutive entries $0,\gamma$, where $\gamma \not=0$, are those where $\gamma=u(b_0+b_1\alpha+\cdots+b_{e-1}\alpha^{e-1})$ and $u\in \F_p$.
Therefore, multiplying any row of $H_e^{(r)}$ by $\alpha^j$, for any $\alpha^j \notin \F_p$, we do not obtain a row of $H_e^{(r)}$. Now, the statement about the kernel is clear from \Cref{lemm:3.1}.

For the rank, we can improve the lower bound given in \Cref{bounds-a}. From the previous paragraph, for any $v\in C_H \backslash \K(C_H)$, the intersection of the linear space over $\F_{p^e}$ generated by $v$ and
$C_H$ is of dimension 1 over $\F_p$. Hence, the number of independent vectors over $F_{p^e}$ generating the $p^{e+t-ek}$ cosets of $\K(C_H)$ over $C_H$ is lower bounded by $e+t-ek=t$ and so $t+k =t+1\leq \rank(C_H)$. Finally, from item $(iv)$ of \Cref{bounds-a} we obtain the statement.
\end{proof}

\begin{example}
We construct a $\gh$ matrix $H(2^2,2)$ such that $C_H$ is an $\F_2$\nobreakdash-additive code
over $\F_{2^2}$ of length $n=2^3$ with $\ker(C_H)=1$ and $\rank(C_H)=4$.
We begin with the $\gh$ matrix $H(2^3,1)$ given by the multiplicative table of $\F_{2^3}$, that is,
\begin{equation}\label{231}
H(2^3,1)=\left(
\begin{array}{cccccccc}
0&0&0&0&0&0&0&0\\
0&1&\omega&\omega^2&\omega^3&\omega^4&\omega^5&\omega^6\\
0&\omega&\omega^2&\omega^3&\omega^4&\omega^5&\omega^6&1\\
0&\omega^2&\omega^3&\omega^4&\omega^5&\omega^6&1&\omega\\
0&\omega^3&\omega^4&\omega^5&\omega^6&1&\omega&\omega^2\\
0&\omega^4&\omega^5&\omega^6&1&\omega&\omega^2&\omega^3\\
0&\omega^5&\omega^6&1&\omega&\omega^2&\omega^3&\omega^4\\
0&\omega^6&1&\omega&\omega^2&\omega^3&\omega^4&\omega^5\\
\end{array}
\right),
\end{equation}
where $\omega$ is a primitive element in $\F_{2^3}$ and $\omega^3=\omega+1$,
Next, we write each entry of (\ref{231}) by using coordinates over $\F_2$ and projecting them over $\F_{2^2}$.   Note that $\bzero=\overline{(0,0,0)}=(0,0)=0$, $\bone=\overline{(1,0,0)}=(1,0)=1$, $\ow=\overline{(0,1,0)}=(0,1)=\alpha$, $\ow^2=\overline{(0,0,1)}=(0,0)=0$, $\ow^3=\overline{(1,1,0)}=(1,1)=\alpha^2$, $\ow^4=\overline{(0,1,1)}=(0,1)=\alpha$, $\ow^5=\overline{(1,1,1)}=(1,1)=\alpha^2$, $\ow^6=\overline{(1,0,1)}=(1,0)=1$, where $\alpha$ is a primitive element in $\F_{2^2}$ and $\alpha^2=\alpha+1$. Finally, by the proof of \Cref{construction_t}, we obtain the following $\gh$ matrix $H(2^2,2)$:
\begin{equation}
H(2^2,2)=\left(
\begin{array}{cccccccc}
0&0&0&0&0&0&0&0\\
0&1&\alpha&0&\alpha^2&\alpha&\alpha^2&1\\
0&\alpha&0&\alpha^2&\alpha&\alpha^2&1&1\\
0&0&\alpha^2&\alpha&\alpha^2&1&1&\alpha\\
0&\alpha^2&\alpha&\alpha^2&1&1&\alpha&0\\
0&\alpha&\alpha^2&1&1&\alpha&0&\alpha^2\\
0&\alpha^2&1&1&\alpha&0&\alpha^2&\alpha\\
0&1&1&\alpha&0&\alpha^2&\alpha&\alpha^2\\
\end{array}
\right).
\end{equation}
\end{example}

\medskip
By \Cref{bounds-a}, for $q=p^2$ (that is, for $e=2$), we have that the rank of an $\F_p$-additive $\gh$ code $C_H$ of length $n=p^2$ (that is, with $t=e=2$) has to be $3$ if the dimension of the kernel is $1$. The next result shows that there exists a code $C_H$ with these parameters for any $p$ prime and $p \not =2$. 
Therefore, in this case, we have that the given lower bound for the rank, once the dimension of the kernel is fixed, coincides with the upper bound. 

Note that, for $q=4$, it is well known that there is only one $\gh$ matrix $H(4,1)$, up to equivalence, which gives the linear code $C_H$ of length $n=4$ over $\F_4$, having rank and dimension of the kernel equal to $2$. 

\begin{proposition}\label{prop:r3k1}
For $t=e=2$ and $p$ an odd prime, there exists a $\gh$ matrix $H(p^2,1)$ such that $C_H$ is an $\F_p$-additive code over $\F_{p^2}$ of length $n=p^2$ with $\ker(C_H)=1$ and $\rank(C_H)=3$.
\end{proposition}
\begin{proof} 
First, we construct a $\gh$ matrix $H(p^2,1)$ and then we prove that it fulfils the conditions of the statement. Let $\vv_1=(0,\omega^0,\omega^1,\ldots,\omega^{p^2-2})$, where $\omega$ is a primitive element in $\F_{p^2}$. Hence, $\vv_1$ has a zero in the first position and $\omega^{i-1}$ in the $(i+1)$th position, for $i\in \{1,2,\ldots,p^2-1\}$. Now, we define $\vv_2$ as the vector having a zero in the first position and $\omega^{ip}$ in the $(i+1)$th position. We construct the matrix $H(p^2,1)$ having as rows all linear combinations over $\F_p$ of $\vv_1$ and $\vv_2$. We permute the rows in order to have the all-zero row as the first one. Note that $H(p^2,1)$ is a matrix of order $p^2$ over $\F_{p^2}$, which has all zeros in the first row and column. 

The vector $\vv_1$ has no two positions with the same value. The same is true for $\vv_2$  (indeed, the coordinates of $\vv_2$ correspond to the image of the Frobenius automorphism $x\rightarrow x^p$). The elements of $\F_{p^2}$ which are in $\F_p$ are of the form $\omega^{\lambda(p+1)}$, for $\lambda\in \{0,1,\ldots,p-1\}$. Hence, the $(i+1)$th position of any row in $H(p^2,1)$ is of the form $\omega^{\lambda(p+1)} \omega^{i-1} +\omega^{\gamma(p+1)}\omega^{ip}$, for  $i\in \{1,2,\ldots,p^2-1\}$ and  $\lambda,\gamma\in \{0,1,\ldots,p-1\}$. The coordinates of any row of $H(p^2,1)$ are all different, otherwise we would have two indexes $i,j$ such that
$$
\omega^{\lambda(p+1)} \omega^{i-1} +\omega^{\gamma(p+1)}\omega^{ip}=\lambda^{\gamma(p+1)} \omega^{j-1} +\omega^{\gamma(p+1)}\omega^{jp}
$$
or, equivalently, $\omega^{\lambda(p+1)-1}(\omega^i-\omega^j)=\omega^{\gamma(p+1)}(\omega^{jp}-\omega^{ip})$. Note that $(\omega^{jp}-\omega^{ip})=(-1)^p (\omega^i-\omega^j)^p$. Therefore, simplifying, there would exist $\delta\in  \{0,1,\ldots,$ $p-1\}$ such that $\omega^{\delta(p+1)-1}=(\omega^i-\omega^j)^{p-1}$. Raising this equality to $p+1$, we obtain  $(p+1)(\delta(p+1)-1)\equiv 0 \pmod{p^2-1}$. Reducing modulo $2\xi$, where $\xi$ is the highest power of $2$ dividing $p+1$, we obtain $p+1\equiv 0 \pmod{2\xi}$, which contradicts the definition of $\xi$. This proves that $H(p^2,1)$ is a $\gh$ matrix. 

Finally, by \Cref{bounds-a}, we just need to prove that $H(p^2,1)$ is nonlinear. We see that multiplying $\vv_1$ by $\omega^p$, we obtain a vector which is not any row of the matrix. It is enough to focus on the second column of the matrix. Assume the contrary, that is, there exist  $\lambda,\gamma\in \{0,1,\ldots,p-1\}$ such that $\omega^p = \omega^{\lambda(p+1)} +\omega^{\gamma(p+1)}\omega^{p}$  or, equivalently, $\omega^{\lambda(p+1)}=(1-\omega^{\gamma(p+1)})\omega^{p}$. Since $\omega^{\lambda(p+1)}$ and $1-\omega^{\gamma(p+1)} \in \F_p$, we would have that $\omega^p\in \F_p$, which is a contradiction. 
\end{proof}

\begin{example} 
We construct a $\gh$ matrix $H(9,1)$ such that $C_H$ is an $\F_3$\nobreakdash-additive code
over $\F_{9}$ of length $n=9$ with $\ker(C_H)=1$ and $\rank(C_H)=3$.
Let $\vv_1=(0,\omega^0,\omega,\omega^2,\omega^3,\omega^4,\omega^5,\omega^6,\omega^{7})$ and $\vv_2=(0,\omega^3,\omega^6,\omega,\omega^4,\omega^7,\omega^2,\omega^5,\omega^{0})$, where $\omega$ is a primitive element in $\F_{3^2}$ and $\omega^3=\omega+2$. By the proof of \Cref{prop:r3k1}, we obtain the matrix $H(9,1)$ as the matrix having as rows all linear combinations over $\F_3$ of $\vv_1$ and $\vv_2$, that is,
\begin{equation}
H(3^2,1)=\left(
\begin{array}{ccccccccc}
0&0&0&0&0&0&0&0&0\\
0& 1& \omega& \omega^2& \omega^3& \omega^4& \omega^5& \omega^6& \omega^7 \\ 
0& \omega^3& \omega^6& \omega& \omega^4& \omega^7& \omega^2& \omega^5& 1 \\ 
0& \omega^4& \omega^5& \omega^6& \omega^7& 1& \omega& \omega^2& \omega^3 \\ 
0& \omega^7& \omega^2& \omega^5& 1& \omega^3& \omega^6& \omega& \omega^4 \\ 
0&  \omega &   \omega^3& 1& \omega^6&  \omega^5& \omega^7& \omega^4& \omega^2\\
0& \omega^5& \omega^7&  \omega^4& \omega^2&   \omega& \omega^3&   1& \omega^6\\
0& \omega^6&  \omega^4& \omega^3& \omega^5& \omega^2&   1& \omega^7&   \omega\\
0& \omega^2&  1& \omega^7&   \omega& \omega^6&  \omega^4& \omega^3& \omega^5\\
\end{array}
\right).
\end{equation}
\end{example}

\medskip
By \Cref{bounds-a}, for $q=p^3$ (that is, for $e=3$), we have that the rank of an $\F_p$-additive $\gh$ code $C_H$ of length $n=p^3$ (that is, with $t=e=3$) is $3$ or $4$ if the dimension of the kernel is $1$. However, by computer search, we found that there is not any $\F_2$-additive $\gh$ matrix $H(8,1)$ with $C_H$ of rank $3$, and they do exist with rank $4$, for example the one given in \Cref{ExH8rank4}. Therefore, the lower bound given by \Cref{bounds-a} is not always tight.  

\begin{example} \label{ExH8rank4}
The following $\gh$ matrix $H(2^3,1)$ over $\F_{2^3}$ generates a $\F_2$-additive $\gh$ code $C_H$ of length $n=2^3$ with rank equal to $4$ and a kernel of dimension $1$, where $\omega$ is a primitive element in $\F_{2^3}$ and $\omega^3=\omega+1$. 
\begin{equation}
H(2^3,1)=\left(
\begin{array}{cccccccc}
0&0&0&0&0&0&0&0\\
0 & \omega^2 & \omega^5 & \omega^3 & 1 & \omega & \omega^6 & \omega^4 \\
0 & \omega^3 & \omega & 1 & \omega^6 & \omega^2 & \omega^4 & \omega^5 \\ 
0 & 1 & \omega^2 & \omega^6 & \omega^4 & \omega^3 & \omega^5 & \omega \\ 
0 & \omega^5 & \omega^6 & a & \omega^2 & \omega^4 & \omega^3 & 1 \\ 
0 & \omega^6 & \omega^3 & \omega^4 & \omega^5 & 1 & \omega & \omega^2 \\
0 & \omega & \omega^4 & \omega^2 & \omega^3 & \omega^5 & 1 & \omega^6 \\ 
0 & \omega^4 & 1 & \omega^5 & \omega & \omega^6 & \omega^2 & \omega^3\\ 
\end{array}
\right).
\end{equation}
\end{example}

For $t=e=4$, we have found that there are $\gh$ matrices $H(p^4,1)$ for $p=3$ and $p=5$ such that $C_H$ are $\F_p$-additive codes with minimum dimension of the kernel $\ker(C_H)=1$ and minimum rank $\rank(C_H)=3$, given by the following two examples. We have checked computationally that the technique used in these two examples does not apply for $p=7$ and $p=11$. 

\begin{example} \label{ExHq3rank3}
Let $H(3^4,1)$ be the matrix having as rows all linear combinations over $\F_3$ of $\vv_1$, $\omega\vv_1$, $\vv_2$ and $\omega\vv_2$, where $\vv_1=(0,\omega^0, \omega^1, \ldots, \omega^i, \ldots, \omega^{79})$,
$\vv_2=(0,\omega^2, \omega^{11}, \ldots, \omega^{2+9i}, \ldots, \omega^{73})$, and $\omega$ is a primitive element in $\F_{3^4}$. We can check that it is a $\gh$ matrix. Clearly, by construction, $C_H$ is a $\F_3$-additive code with $\rank(C_H)=3$. By \Cref{lemm:3.1}, we have that $\ker(C_H)=1$. 
\end{example}

\begin{example} \label{ExHq5rank3}
Let $H(5^4,1)$ be the matrix having as rows all linear combinations over $\F_5$ of $\vv_1$, $\omega\vv_1$, $\vv_2$ and $\omega\vv_2$, where $\vv_1=(0,\omega^0, \omega^1, \ldots, \omega^i, \ldots, \omega^{623})$,
$\vv_2=(0,\omega^6, \omega^{31}, \ldots, \omega^{6+25i}, \ldots, \omega^{605})$, and $\omega$ is a primitive element in $\F_{5^4}$. We can check that it is a $\gh$ matrix. Clearly, by construction, $C_H$ is a $\F_5$-additive code with $\rank(C_H)=3$. By \Cref{lemm:3.1}, we have that $\ker(C_H)=1$. 
\end{example}

\section{Combining different constructions}
\label{sec:CombConst}

In this section, we use the constructions of $\F_p$-additive $\gh$ codes given in \Cref{sec:Kronecker,sec:constKernel1}, to show the existence of such codes having different ranks between the lower and upper bounds found in \Cref{sec:bounds}, for a fixed dimension of the kernel. We also see that it is only possible to construct codes for all allowable pairs rank and dimension of the kernel when $e=2$, by using the above constructions. For $e\geq 3$, mainly it is still necessary to prove the existence of $\F_p$-additive $\gh$ codes with a kernel of dimension $1$ and not having the maximum rank.  

\medskip
First, in the next theorem, for $q=p^e$, $p$ prime and any $t> e>1$, we prove that there is an $\F_p$-additive $\gh$ code $C_H$ over $\F_q$ of length $n=p^t$ with $\ker(C_H)=k$  for each possible value of $k$ given by item $(iv)$ of \Cref{bounds-a}. These codes are the ones having the maximum rank, that is, satisfying the upper bound given by item $(iii)$ of \Cref{bounds-a}. This proves that this upper bound for the rank, once the dimension of the kernel is fixed, is tight for all cases with $t>e>1$.  

Note that when $t=e>1$, there is an $\F_p$-additive $\gh$ matrix $H(q,1)$, given by the multiplicative table
of $\F_q$, so the corresponding $\gh$ code $C_H=C_{S_q}$ of length $n=p^e=q$ is linear over $\F_q$ and $\ker(C_H)=\rank(C_H)=2$ \cite{tonchev}. According to \Cref{bounds-a}, in this case, there could be $\F_p$-additive $\gh$ codes $C_H$ having $\ker(C_H)=1$ and $\rank(C_H)\in \{3,\ldots,1+e \}$. By \Cref{prop:r3k1}, for $t=e=2$, there exist such codes having rank $3$. However, it is still an open problem to prove their existence when $q=p^e$ and $e\geq 3$ (except for $q=2^3$, $q=3^4$ and $q=5^4$; by \Cref{ExH8rank4,ExHq3rank3,ExHq5rank3}, respectively), even in general for $\gh$ codes which are not necessarily $\F_p$-additive. 
These are connected to Latin squares of order $q-1$ and we could use this approach to construct them.

\begin{theorem} \label{kernel_additiu} For $q=p^e$, $p$ prime, and any $t> e>1$, there exists an  $\F_p$\nobreakdash-additive $\gh$ code $C_H$ over $\F_q$ of length $n=p^t$  with $\ker(C_H)=k$
if and only if
\begin{enumerate}[(i)]
\item $k\in \{1, \ldots, \lfloor t/e \rfloor \}$ when $e \nmid t$,
\item $k\in \{1, \ldots, t/e +1 \}$, otherwise.
\end{enumerate}
Moreover, $\rank(C_H)=1+t-(e-1)(k-1)$.
\end{theorem}

\begin{proof}
The general proof is by induction over $t>e$, in steps of $e$.
In this sense, the first point is to show the existence of  $\F_p$\nobreakdash-additive $\gh$ codes with the following parameters:
\begin{align*}
k=1             & \textrm{ when } t\in \{e+1,\ldots,2e-1 \},\\
k \in\{1,2,3\}  & \textrm{ when } t= 2e.
\end{align*}
When $t=2e$, by \Cref{coro:2.2}, the $\F_p$\nobreakdash-additive $\gh$ code $C_H$ corresponding to 
$S^2=S_q \oplus S_q$ has $k=3$ and $\rank(C_H)=1+t-(e-1)(k-1)=3$. By \Cref{constructionSwitch}, there exists an  $\F_p$\nobreakdash-additive  $\gh$ code $C_H$ over $\F_q$ with $k=2$ and $\rank(C_H)=1+t-(e-1)(k-1)=2+e$. From \Cref{construction_t}, the existence of $\F_p$\nobreakdash-additive  $\gh$ codes with $k=1$ and maximum rank is assured for all $t>e>1$.

By \Cref{coro:2.2} and \Cref{lem:Kro1}, we can recursively construct $\F_p$-additive $\gh$ codes $C_H$ by using the Kronecker sum construction, $H=S_q \oplus B$, where $B$ is an $\F_p$-additive $\gh$  matrix of size $p^{t'}=p^{t-e}$ constructed in the previous step. Note that if $C_B$ of length $p^{t'}$ has a kernel of dimension $k'$ and maximum rank $r'=1+t'-(e-1)(k'-1)$,
then $C_H$ of length $p^t$ with $t=t'+e$ has a kernel of dimension $k=k'+1$ and rank $r=r'+1$. Therefore, $r=1+t'-(e-1)(k'-1) +1 = 1+t-(e-1)(k-1)$ and $C_H$ has maximum rank. This construction covers all the values of $k$ in the statement, except $k=1$. However, \Cref{construction_t} assures the existence of $\F_p$\nobreakdash-additive $\gh$  codes with $k=1$ and maximum rank for all $t>e>1$. Therefore, the existence for all given parameters is proved. 

Finally, by \Cref{bounds-a} and \Cref{lemm:maxkernel}, we have that these are the
only possibilities for the dimension of the kernel of an $\F_p$-additive $\gh$ code.
\end{proof}

When $e=2$, by \Cref{coro:bounds-a}, we have that for each possible dimension of the kernel, there is only one possible rank. Therefore, when $t>e=2$, the above theorem covers all possible pairs $(r,k)$, where $r$ is the rank and $k$ the dimension of the kernel of the $\F_p$-additive $\gh$ code. Moreover, since \Cref{prop:r3k1} covers the case when $t=e=2$, we have the following corollary.

\begin{corollary} \label{cor:3.9} For $q=p^2$, $p$ prime, and any $t\geq 2$, there exists an  $\F_p$\nobreakdash-additive $\gh$ code $C_H$ over $\F_q$ of length $n=p^t$  with $\ker(C_H)=k$ and $\rank(C_H)=r$ 
if and only if $r=t+2-k$ and $k$ satisfies 
\begin{enumerate}[(i)]
\item $k\in \{1, \ldots, \lfloor t/e \rfloor \}$ when $e \nmid t$,
\item $k\in \{1, \ldots, t/e +1 \}$, otherwise.
\end{enumerate}
\end{corollary}

\begin{proof}
Straightforward by \Cref{bounds-a,coro:bounds-a,prop:r3k1,kernel_additiu}.
\end{proof}

\begin{example}
For $q=p^2$, the second column in Table~\ref{tableF4} gives all possible values for the rank and dimension of the kernel of $\F_p$-additive $\gh$ codes over $\F_{p^2}$ of length $n=p^t$ with $2\leq t \leq 7$, given by \Cref{bounds-a,coro:bounds-a}. By \Cref{cor:3.9}, for each one of these pairs, there exists an $\F_p$-additive  $\gh$ code over $\F_{p^2}$ having these parameters.

\begin{table}[htp]%
\begin{center}
\begin{tabular}{ccc}
\toprule%
$t$ &  $(\rank(C_H), \ker(C_H))$ & $\ker_p(C_H)=\rank_p(C_H)$ \\ \toprule
2   &   (3,1) (2,2)    & 2 \\
3   &   (4,1)    & 2.5\\
4   &   (5,1) (4,2) (3,3) & 3 \\
5   &   (6,1) (5,2)  & 3.5\\
6   &   (7,1) (6,2) (5,3) (4,4) & 4 \\
7   &   (8,1) (7,2) (6,3) & 4.5 \\ \bottomrule
\end{tabular}\caption{Parameters $(\rank(C_H), \ker(C_H))$ for all  $\F_p$-additive $\gh$ codes $C_H$ over $\F_{p^2}$
of length $n=p^t$ with $2\leq t\leq 7$.}\label{tableF4}
\end{center}
\end{table}
\end{example}

\begin{example}
For $q=4$, that is when $p=2$ and $e=2$, we can take into account some already known results on the classification of $\gh$ matrices.
\begin{itemize}
\item If $n=4$ ($t=2$), there is only one $\gh$ matrix $H(4,1)$ over $\F_4$ having $\rank(C_H)=\ker(C_H)=2$.
Therefore, $C_H$ is linear over $\F_4$, so $\F_2$-additive. Actually, $H(4,1)$ corresponds to the Sylvester $\gh$ matrix $S^1=S_4$.

\item If $n=8$ ($t=3$), there is only one $\gh$ matrix $H(4,2)$ having $\rank(C_H)=4$ and $\ker(C_H)=1$. Therefore,
$C_H$ is nonlinear over $\F_4$. However, since $\rank_2(C_H)=\ker_2(C_H)=2.5$, it is $\F_2$\nobreakdash-additive.

\item If $n=16$ ($t=4$), it is known that there are 226 nonequivalent $\gh$ matrices $H(4,4)$  over $\F_4$ \cite{Harada}. \Cref{tableN16} shows the ranks and dimensions of the kernel
of the corresponding codes $C_H$. Moreover, for each case, it also gives the value $Na / N$, where $N$ is the number of nonequivalent codes and $Na$ the number of such codes that are $\F_2$-additive.  
\begin{table}[ht]
\begin{center}
\begin{tabular}{r | c c c c c c}
   & \multicolumn{6}{c}{$\rank(C_H)$} \\
$\ker(C_H)$  & 3 & 4 & 5 & 6 & 7 & 8 \\
\hline
3 & 1/1 &   &   &     \\
2 &   & 5/7  & 0/8 &     \\
1 &   & 0/3  & 38/92  & 0/55 & 0/57 & 0/3    \\
\end{tabular}\caption{Number of nonequivalent $\F_2$-additive $\gh$ codes of length $16$ over $\F_4$ versus the total number, for each possible rank and dimension of the kernel.}\label{tableN16}
 \end{center} \end{table}
\end{itemize}
\end{example}

Recall that the more general switching construction, given by \Cref{constructionSwitch2}, allows for the construction of $\F_p$-additive $\gh$ codes having different ranks and dimensions of the kernel, when $t$ is a multiple of $e$ and $k>1$. Now, we combine this construction with the Kronecker sum construction to cover more cases, proving the existence of $\F_p$-additive $\gh$ codes $C_H$ over $\F_q$ of length $n=p^t$, where $q=p^e$, when $t$ is {\it not} a multiple of $e$ and $k>1$.

\begin{proposition} \label{constructionSwitch3} 
For $q=p^e$, $p$ prime, and any $t$ not a multiple of $e>1$,
there exists an  $\F_p$\nobreakdash-additive $\gh$ code $C_H$ over $\F_q$ of length $n=p^t$ with $\rank(C_H)=r$ and $\ker(C_H)=k$, for all $k \in \{2,\ldots, h=\lfloor t/e \rfloor\}$ and
$$r \in \{2h-k+t+e-he,\ldots, 1+t-(e-1)(k-1)\}.$$
\end{proposition}

\begin{proof}
Take $q=p^e$, $p$ prime, and any $t>e$ not a multiple of $e>1$. Set $t=he+h'$, where $h=\lfloor t/e \rfloor$. Hence, we can write $t=(h-1)e+(e+h')$. 

If $h=k$, then $r=2h-k+t+e-he=1+t-(e-1)(k-1)$, so the code has the maximum rank $r$, and its existence is given by \Cref{kernel_additiu}. 
Note that if $h=2$, then $k=2$, so we can focus on the case where $h\geq 3$. 

If $h\geq 3$, by \Cref{constructionSwitch2}, there exists a $\gh$ code $D$ of length $p^{(h-1)e}$ with $\ker(D)=k$ for all $k \in \{2,\ldots, h-1\}$ and $\rank(D) \in \{2(h-1)-k+2, \ldots, h+(h-k)(e-1)\}$. Also, by \Cref{construction_t}, there exists a $\gh$ code $E$ of length $p^{e+h'}$ with $\ker(E)=1$ and $\rank(E)= e+h'+1$. Then, by \Cref{lem:Kro1}, using the Kronecker sum construction with $D$ and $E$, we obtain an $\F_p$\nobreakdash-additive $\gh$ code $C$ over $\F_q$ of length $n=p^t$ with $\ker(C)=k$ and $\rank(C)=r$, for all $2\leq k \leq h -1$ and
$2(h-1)-k+2+e+h'+1-1 \leq r \leq h+(h-k)(e-1)+e+h'+1-1$, or equivalently, $ 2h-k+t+e-he \leq  r \leq 1+t-(e-1)(k-1)$. Therefore, the result follows.
\end{proof}

Note that the $\F_p$-additive $\gh$ codes constructed from the more general switching construction, given by \Cref{constructionSwitch2},  do not cover all possible pairs $(r,k)$, where $r$ is the rank and $k$ the dimension of the kernel, when $t$ is a multiple of $e$ and $k>1$. The upper bounds in \Cref{bounds-a,constructionSwitch2} coincide since $h+1+(h-k+1)(e-1)=1+t-(e-1)(k-1)$ if $t=he$. However, the lower bounds do not coincide in general. The smallest case where both propositions disagree is for $e=3$, $t=9$, $h=3$ and $k=2$. By using \Cref{constructionSwitch2}, we know that we can construct $\F_p$-additive $\gh$ codes with these parameters having rank $r$ for all $r\in \{6,7,8\}$. However, from \Cref{bounds-a}, we have that $r\in \{5,6,7,8\}$, and it is not known whether there is a code having rank $r=5$.

As we just noted for the case when $t$ is a multiple of $e$ and $k>1$, the codes constructed in \Cref{constructionSwitch3} when $t$ is not a multiple of $e$  do not cover all possible pairs $(r,k)$ given by \Cref{bounds-a}. Again, the upper bounds coincide, but not the lower bounds. The smallest case where they disagree is for $e=3$, $t=7$, $h=2$ and $k=2$. By using \Cref{constructionSwitch3}, there exist $\F_p$-additive $\gh$ codes with these parameters having rank $r=6$. However, from \Cref{bounds-a}, we have that $r\in \{4,5,6\}$, and the existence of the case with $r \in \{4,5\}$ is not known. 

\begin{theorem} \label{theo:resum}
For $q=p^e$, $p$ prime, and any $t>e>1$,
there exists an  $\F_p$\nobreakdash-additive $\gh$ code $C_H$ over $\F_q$ of length $n=p^t$ with $\ker(C_H)=k$ and $\rank(C_H)=r$, for all $k \in \{2,\ldots, h=\lfloor t/e \rfloor\}$ and
$$r \in \{l_k,\ldots, 1+t-(e-1)(k-1)\},$$ where
\begin{equation*}
  l_k =
    \begin{cases}
      2h-k+2  & \text{if $t$ is multiple of $e$},\\
      2h-k+t+e-he,  & \text{otherwise}.
    \end{cases}       
\end{equation*}
\end{theorem}

\begin{proof}
Straightforward from \Cref{constructionSwitch2,constructionSwitch3}.
\end{proof}

\begin{corollary} \label{cor:3.10} For $q=p^e$, $p$ prime, and any $t > e>1$, there exists an $\F_p$\nobreakdash-additive $\gh$  code $C_H$ over $\F_q$ of length $n=p^t$ with $\ker(C_H)=k$ and $\rank(C_H)=r$, if and only if
\begin{enumerate}[(i)]
\item  $r \in \{t/e + 2, \ldots, t/e + e \}$ when $k=t/e$ ($t$ is a multiple of $e$);
\item  $r=1+t/e$ when $k=1+t/e$ ($t$ is a multiple of $e$).
\end{enumerate}
\end{corollary}
\begin{proof}
By \Cref{theo:resum}, 
we have that there exists an $\F_p$\nobreakdash-additive $\gh$  code $C_H$ over $\F_q$ of length $n=p^t$ 
with $\ker(C_H)=h$ and $\rank(C_H)=r$ for all $r \in \{h + 2, \ldots, h + e \}$, where $h=t/e$. 
In this case, the value $h+2$ coincides with the lower bound  $\lceil\frac{ e+t-t/e}{e-1} \rceil$, given in \Cref{bounds-a}.
Note that $\lceil\frac{ e+t-t/e}{e-1} \rceil =\lceil\frac{ e-1+1+t-t/e}{e-1} \rceil =1 + \lceil \frac{1}{e-1} \rceil +t/e=
t/e+2=h+2$. Similarly, the value $h+e$ is equal to the upper bound given by \Cref{bounds-a}, since  $1+t-(e-1)(t/e-1)=t/e+e=h+e$.
\end{proof}

\begin{example}
In \Cref{tableF8bounds} of \Cref{ex:e3}, all possible values for the rank of $\F_p$-additive $\gh$ codes of length $n=p^t$ with $2\leq t\leq 12$, once the dimension of the kernel is given, are shown. By \Cref{theo:resum}, for each one of these values, except for the ones in bold type, there exists an $\F_p$-additive $\gh$ code having these parameters.  

As it is noticed in \Cref{cor:3.10}, we can also see in \Cref{tableF8bounds} that when $t=3h$ with $t>3$, if $k=h$ or $k=h+1$, we can construct an $\F_p$-additive $\gh$ code $C_H$ with $\ker(C_H)=k$ and $\rank(C_H)=r$ for all possible values of $r$ between the bounds given by \Cref{bounds-a}.
\end{example}

By using the $\F_p$\nobreakdash-additive $\gh$ codes over $\F_{p^4}$ of length $n=p^4$ (with $p=3$ and $p=5$) and a kernel of dimension 1, given in \Cref{ExHq3rank3,ExHq5rank3}, along with the Kronecker sum construction, we show the existence of  $\F_p$\nobreakdash-additive $\gh$ codes over $\F_{p^4}$ with greater length, kernel of dimension 1 and different ranks.  

\begin{proposition}
For $q=p^4$, with $p=3$ or $p=5$, and any $t\geq 4$,
there exists an  $\F_p$\nobreakdash-additive $\gh$ code $C_H$ over $\F_q$ of length $n=p^t$ with $\ker(C_H)=1$ and $\rank(C_H)=t+1-2i $, for all $i \in \{0,\ldots, \lceil t/4 \rceil -2 \}$.
\end{proposition}

\begin{proof}
First, note that the upper bound for the rank is $t+1$ when the dimension of the kernel is $1$.
For $t \in \{5,6,7,8\}$, we have that $r=t+1$ and the statement is true by  \Cref{kernel_additiu}.
Assume that it is true for $t'\in \{4h'-3, 4h'-2, 4h'-1,4h' \}$, $h'\geq 2$. That is, by induction hypothesis, there exist an $\F_p$\nobreakdash-additive $\gh$ code $C_D$ with $\ker(C_D)=1$ and $\rank(C_D)=t'+1-2j$ for all $j \in \{0,\ldots,h'-2 \}$. 
By \Cref{ExHq3rank3,ExHq5rank3}, there exists an $\F_p$\nobreakdash-additive $\gh$ code $C_E$ with $\ker(C_E)=1$ and $\rank(C_E)=3$.  By \Cref{lem:Kro1}, applying the Kronecker sum construction to the corresponding $\gh$ matrices $D$ and $E$, there exist  $\F_p$\nobreakdash-additive $\gh$ codes $C_H$ of length $n=p^t$, with $t=t'+4$ and $h=h'+1$, having $\ker(C_H)=\ker(C_D)+\ker(C_E)-1=1$ and
$\rank(C_H)=\rank(C_D)+\rank(C_E)-1=t'+1-2j+3-1$ for all $j \in \{0,\ldots, h'-2\}$,
or equivalently, $\rank(C_H)=t-4+1-2j+2=t+1-2(j+1)=t+1-2i$ for all $i \in  \{1,\ldots, h-2\}$. 
Finally, again by \Cref{kernel_additiu}, there is a code $C_H$ for $i=0$, and the result follows. 
\end{proof}

Finally, we show that if the existence of $\F_p$\nobreakdash-additive $\gh$ codes with dimension of the kernel 1 and any rank between the given lower and upper bounds is proved, then we would have the existence of any such code with rank $r$ and kernel of dimension $k$ for any possible pair $(r,k)$.   

\begin{theorem} \label{prop:3.12} Let $q=p^e$, $p$ prime, and $e>1$. 
If there  exist an $\F_p$\nobreakdash-additive $\gh$ code $C_H$ over $\F_q$ of length $n=p^t$, $t>e$, with $\ker(C_H)=1$ and $\rank(C_H)=r$ for all $r \in \{\lceil \frac{e+t-1}{e-1} \rceil,\ldots, 1+t \}$, then 
there exists an $\F_p$\nobreakdash-additive $\gh$ code $C_H$ over $\F_q$ of length $n=p^t$ with $\ker(C_H)=k$ and $\rank(C_H)=r$, for all $k\in \{2,\ldots, \lfloor t/e \rfloor \}$ and $r \in \{\lceil \frac{e+t-k}{e-1}\rceil,\ldots, 1+t +(e-1)(k-1)\}$.
\end{theorem}
\begin{proof}
It follows from the same arguments as in the proof of \Cref{kernel_additiu}, by induction over $t$, in steps of $e$.  Note that, for the lower bound of the rank, by induction hypotheses, there exists an $\F_p$\nobreakdash-additive $\gh$ code of length $p^{t-e}$ with kernel of
dimension $k-1$ and rank $\lceil \frac{e+t-e-(k-1)}{e-1}\rceil$. After applying the Kronecker sum construction, the new $\F_p$\nobreakdash-additive $\gh$ code $C_H$ of length $2^t$ would have $\ker(C_H)=k$ and $\rank(C_H)=\lceil \frac{e+t-e-(k-1)}{e-1}\rceil +1=\lceil \frac{e+t-k}{e-1}\rceil$.
\end{proof}

\section{Self-orthogonal $\mathbf{\F_p}$-additive $\gh$ codes over $\mathbf{\F_{p^2}}$ and quantum codes}
\label{sec:quantum}

The question of finding quantum-error correcting codes is transformed into the question of finding additive codes over a finite field which are self-orthogonal with respect to a certain trace inner product \cite{AK01,Cal98}. In this section, we see that the $\F_p$-additive $\gh$ codes over $\F_q$ with $q=p^2$ constructed in the previous sections are self-orthogonal, so they can be used to produce quantum codes. 

\medskip
For codes of length $n$ over $\F_{p^2}$, there are other well known inner products besides the Euclidean inner product. Let $\vv=(v_1,\ldots,v_n)$ and $\vw=(w_1,\ldots,w_n) \in \F_{p^2}^n$.  The {\it Hermitian inner product}, defined by
$[\vv,\vw]_H = \sum_{i=1}^n v_i w_i^p$ and the {\it trace Hermitian inner product}, which is used for $\F_p$\nobreakdash-additive codes over $\F_{p^2}$,
that is $\langle \vv, \vw \rangle = \sum_{i=1}^n (v_i w_i^p - v_i^p w_i)$ or $\langle \vv, \vw \rangle = \beta\sum_{i=1}^n (v_i w_i^p - v_i^p w_i)$ depending on
whether $p$ is even or odd, respectively~\cite{slo}, where $\beta=\omega^{(p+1)/2}$ and $\omega$ is a primitive element in $\F_{p^2}$. 
For short, we will call it {\it additive inner product}. We denote the orthogonal code defined by the additive inner product as $C^\perp$.

\begin{lemma}\label{solambda}
 Let $\vv,\vw$ be  two rows of a $\gh$ matrix $H(p^2,\lambda)$ with $p\not=2$.  Then $\langle \vv,\vw \rangle =0$.
\end{lemma}
\begin{proof}
Since the nonzero elements of $\F_{p}$ are the roots of the polynomial $x^{p-1}-1$, we see that $\sum_{i=0}^{p-2} \alpha^i=0$, where $\alpha \in \F_{p}$ is a primitive element and $p\not= 2$.

Now, let $\omega$ be a primitive element in $\F_{p^2}$ and take $\alpha=\omega^{p+1}$ which is a primitive element in $\F_p$. Then
\begin{equation}\label{suma}
\sum_{j=0}^{p^2-2} (\omega^j)^{p+1} = \sum_{i=0}^{p-2} \sum_{t=1}^{p+1}(\omega^{i+t(p-1)})^{p+1} = (p+1)\sum_{i=0}^{p-2}(\omega^{i})^{p+1} =(p+1)\sum_{i=0}^{p-2}\alpha^{i}=0.
\end{equation}

For any two rows $\vv=(v_1,\ldots,v_n)$ and $\vw=(w_1,\ldots,w_n)$ in  a $\gh$ matrix $H(p^2,\lambda)$, where $n=p^2\lambda$, we have that $\sum_{i=1}^n (v_i - w_i)^{p+1} = \sum_{i=1}^n v^{p+1}_i = \sum_{i=1}^n w^{p+1}_i =\lambda \sum_{j=0}^{p^2-2} (\omega^j)^{p+1} =0$ by (\ref{suma}).
Then we have
\begin{multline}
\sum_{i=1}^n (v_i - w_i)^{p+1} = \sum_{i=1}^n (v^p_i-w^p_i)(v_i-w_i)= \\
\sum_{i=1}^n (v^{p+1}_i +w^{p+1}_i -(v^p_iw_i+w^p_iv_i))= - \sum_{i=1}^n (v^p_iw_i+w^p_iv_i)=0\\
\end{multline}
Therefore, if $p$ is even, then  $\langle \vv,\vw \rangle =0$. Otherwise, $0=-\beta \sum_{i=1}^n (v^p_iw_i+w^p_iv_i)= \sum_{i=1}^n ((\beta v_i)^p w_i- w^p_i (\beta v_i))=\langle \beta \vv,\vw \rangle =\beta \langle \vv,\vw \rangle$ and we also have that $\langle \vv,\vw \rangle =0$. Note that since $\beta=\omega^{(p+1)/2}$ when $p$ is odd, $\beta^p=-\beta$.
\end{proof}

 \begin{lemma}\label{solambda2}
 Let $\vv,\vw$ be  two different rows of a $\gh$ matrix $H(4,\lambda)$.  Then $\langle \vv,\vw \rangle  \equiv \lambda \pmod{2}.$  If $\vv = \vw$ then $\langle \vv,\vw \rangle \equiv 0 \pmod{2}.$
\end{lemma}
\begin{proof}
We note that for a vector $\vv=(v_1,\ldots,v_n)$, $\sum_{i=1}^n v_i^3$ is the weight of $\vv$ since each nonzero element cubed is 1. Let $\vv \neq \vw$.
Note also that $\wt(\vv) = \wt(\vw) = \wt(\vv + \vw) = n- \lambda$ by the properties of a $\gh$ matrix, where $n=4\lambda$ is the order.
Then we have
\begin{eqnarray*}
\sum_{i=1}^n (v_i + w_i)^3 &=& \sum_{i=1}^n v_i^3 + 3 \sum_{i=1}^n v_i^2w_i + 3 \sum_{i=1}^n v_i w_i^2 + \sum_{i=1}^n w_i^3 \\
\wt(\vv +\vw) &=& \wt(\vv) + 3\langle \vv,\vw \rangle + \wt(\vw).
\end{eqnarray*}
Hence,  $\langle \vv,\vw \rangle =n-\lambda =3\lambda$ and so $\langle \vv,\vw \rangle  \equiv \lambda \pmod{2}$.

If $\vv = \vw$, then $\wt(\vv+\vw) =0$ and by the same computation as before we have $\langle \vv,\vw \rangle \equiv 0 \pmod{2}.$
\end{proof}

\begin{proposition}
Let $H(p^2,\lambda)$ be a $\gh$ matrix, where $\lambda$ is even when $p=2$. Then $F_H \subset F_H^\perp$ is a  self-orthogonal additive code.
\end{proposition}
\begin{proof}
By Lemmas~\ref{solambda} and \ref{solambda2}, we have that any two vectors are orthogonal with respect to the additive inner product.  
\end{proof}

A $q$-ary quantum code of length $n$ and size $K$ is a $K$-dimensional subspace of a $q^n$-dimensional Hilbert space,
and is denoted by $[[n, k, d]]_q$, where $k=\log_q K$ and $d$ the minimum distance.

\begin{theorem}
Let $H(p^2,\lambda)$ be a $\gh$ matrix of order $n=p^2\lambda=p^t$ over $\F_{p^2}$, where $\lambda$ is even when $p=2$. Then $C_H$ is and $\F_p$-additive self-orthogonal code, containing $p^{t+2}$ codewords, 
such that there are no vectors of weight less than 3 in $C_H^\perp \backslash C$. Then, we have a pure additive quantum-error-correcting code with parameters $[[p^t, p^t-(t+2)/2, 3]]_{p^2}$.
\end{theorem}


It is worth mentioning that all the results in this section are also true for any integer $p>1$, not necessarily a prime number. Moreover, the invariants rank and dimension of the kernel, used in this paper for $\F_p$-additive $\gh$ codes, could be used in general to classify $\F_p$-additive and quantum codes.


\end{document}